\documentclass[12pt]{article}

\usepackage[T1]{fontenc}
\usepackage{kpfonts}

\usepackage{setspace}
\usepackage[margin=1.25in]{geometry}

\usepackage[dvipsnames]{xcolor}
\usepackage{pdfpages}
\usepackage{float}
\usepackage{multirow}
\usepackage{caption}
\usepackage{subcaption}
\usepackage{epigraph}

\usepackage{mathtools}
\usepackage{amsthm}
\usepackage{aliascnt}
\usepackage{accents}
\usepackage{dutchcal}
\usepackage{bbm}

\usepackage{enumitem}
\usepackage{sgame}
\usepackage{tikz}
\usepackage{tikz-cd}
\usetikzlibrary{calc,shapes,arrows}
\usepackage{tcolorbox}
\usepackage{listings}
\usepackage[normalem]{ulem}

\usepackage[round]{natbib}
\DeclareMathOperator*{\argmax}{arg\,max}

\newcommand{\E}{\mathbb{E}}

\renewcommand{\D}{\mathcal{D}}

\newtheorem{theorem}{Theorem}[section]

\newaliascnt{proposition}{theorem}

\aliascntresetthe{proposition}

\newaliascnt{lemma}{theorem}
\newtheorem{lemma}[lemma]{Lemma}
\aliascntresetthe{lemma}

\newaliascnt{corollary}{theorem}
\newtheorem{corollary}[corollary]{Corollary}
\aliascntresetthe{corollary}

\newaliascnt{claim}{theorem}

\aliascntresetthe{claim}

\theoremstyle{definition}

\newaliascnt{definition}{theorem}
\newtheorem{definition}[definition]{Definition}
\aliascntresetthe{definition}

\newaliascnt{example}{theorem}

\aliascntresetthe{example}

\newaliascnt{assumption}{theorem}

\aliascntresetthe{assumption}

\newaliascnt{condition}{theorem}

\aliascntresetthe{condition}

\newaliascnt{question}{theorem}

\aliascntresetthe{question}

\newaliascnt{remark}{theorem}

\aliascntresetthe{remark}

\newaliascnt{remarks}{theorem}

\aliascntresetthe{remarks}

\newaliascnt{aside}{theorem}

\aliascntresetthe{aside}

\newaliascnt{note}{theorem}

\aliascntresetthe{note}

\usepackage{thmtools}
\usepackage{thm-restate}

\usepackage{hyperref}

\hypersetup{
    colorlinks=true,
    linkcolor=OrangeRed,
    filecolor=Thistle,
    urlcolor=Thistle,
    citecolor=Thistle,
}

\usepackage[nameinlink]{cleveref}

\crefname{theorem}{theorem}{theorems}
\Crefname{theorem}{Theorem}{Theorems}
\crefname{proposition}{proposition}{propositions}
\Crefname{proposition}{Proposition}{Propositions}
\crefname{lemma}{lemma}{lemmas}
\Crefname{lemma}{Lemma}{Lemmas}
\crefname{corollary}{corollary}{corollaries}
\Crefname{corollary}{Corollary}{Corollaries}
\crefname{claim}{claim}{claims}
\Crefname{claim}{Claim}{Claims}
\crefname{definition}{definition}{definitions}
\Crefname{definition}{Definition}{Definitions}
\crefname{example}{example}{examples}
\Crefname{example}{Example}{Examples}
\crefname{assumption}{assumption}{assumptions}
\Crefname{assumption}{Assumption}{Assumptions}

\makeatletter
\let\cref@old@isrefconsecutive\cref@isrefconsecutive
\def\cref@isrefconsecutive#1#2{%
  \begingroup
    \def\cref@assumptiontype{assumption}%
    \cref@gettype{#1}{\cref@typea}%
    \ifx\cref@typea\cref@assumptiontype
      \endgroup
      \@cref@refconsecutivefalse
    \else
      \endgroup
      \cref@old@isrefconsecutive{#1}{#2}%
    \fi
}
\makeatother

\crefname{condition}{condition}{conditions}
\Crefname{condition}{Condition}{Conditions}
\crefname{question}{question}{questions}
\Crefname{question}{Question}{Questions}
\crefname{remark}{remark}{remarks}
\Crefname{remark}{Remark}{Remarks}
\crefname{remarks}{remarks}{remarks}
\Crefname{remarks}{Remarks}{Remarks}
\crefname{aside}{aside}{asides}
\Crefname{aside}{Aside}{Asides}
\crefname{note}{note}{notes}
\Crefname{note}{Note}{Notes}
\crefname{appendix}{appendix}{appendices}
\Crefname{appendix}{Appendix}{Appendices}

\newcommand{\secref}[1]{\hyperref[#1]{\S\ref*{#1}}}

\definecolor{backcolour}{rgb}{0.63, 0.79, 0.95}
\lstdefinestyle{mystyle}{
  backgroundcolor=\color{backcolour},
  basicstyle=\ttfamily\footnotesize,
  breakatwhitespace=false,
  breaklines=true,
  captionpos=b,
  keepspaces=true,
  numbers=left,
  numbersep=5pt,
  showspaces=false,
  showstringspaces=false,
  showtabs=false,
  tabsize=2
}
\begin{document}
\title{To Gamble, Perchance to Grow}
\author{Mark Whitmeyer\thanks{Arizona State University. Email: \href{mailto:mark.whitmeyer@gmail.com}{mark.whitmeyer@gmail.com}.}}
\date{\today}
\maketitle

\begin{abstract}
I study transformations of returns in the growth-optimal (Kelly) portfolio problem. In the one-safe-one-risky-asset problem, a return transform \(f\) universally produces a more conservative portfolio if and only if \(f\) is concave and strictly increasing and \(r/f\) is convex. As a corollary, I characterize comparative risk aversion for a rationally-inattentive agent: a more risk-averse agent is one who is sufficiently more risk averse in the \citet{Pratt1964} sense.
\end{abstract}

\section{Introduction}\label{sec:intro}

The more risk averse have utilities that are more concave: \(u \mapsto \varphi \circ u\) with \(\varphi\) concave and strictly increasing \citep{Pratt1964}. The strongly more risk averse have (rescaled) utilities that are additively more concave: \(u \mapsto \varphi \circ u\) with \(\lambda u - \varphi \circ u\) convex and strictly increasing for some \(\lambda > 0\) \citep{Ross1981}. The more prudent have utilities whose slopes are more convex: \(u \mapsto \varphi \circ u\) with \(\varphi'\) convex and strictly increasing \citep{kimball1990}. What transformations of returns make an investor in the growth-optimal portfolio problem \citep{Kelly1956} apparently more risk averse?

To elaborate, I consider an investor who maximizes expected \(\log\) growth \`{a} la \citet{Kelly1956}.\footnote{\citet{kellysurvey} surveys this literature.} For a safe gross return \(s>0\) and a risky gross return \(R>0\), the one-safe-one-risky growth problem is, letting \(\alpha\) denote the risky share,
\[\max_{\alpha\in \left[0,1\right]}\E\log((1-\alpha)s+\alpha R).\]
A transformation \(f\colon\left(0,\infty\right)\to\left(0,\infty\right)\) sends every asset payoff \(r\) to \(f(r)\), so that the transformed problem is
\[\max_{\alpha\in \left[0,1\right]}\E\log((1-\alpha)f(s)+\alpha f(R)).\]
I interpret \citet{Yaari1969}'s behavioral definition of comparative risk aversion as saying that a more risk-averse investor would always have a lower risky share. My question is: which transformations make the Kelly investor universally more conservative in this manner, for every risk-free return and every finite risky-return distribution?

\medskip

\noindent \Cref{thm:main-deterministic}. A transformation raises safe demand universally if and only if
\[
        f\text{ is strictly increasing and concave}
        \qquad\text{and}\qquad
        r\mapsto \frac{r}{f(r)}\text{ is convex}.
\]
Ordinary concavity is not enough. We also need the second harmonic curvature restriction.\footnote{The mirror characterization is the obvious inverse: a transformation raises risky demand universally if and only if \(f\) is strictly increasing and convex and \(r/f(r)\) is concave.}

Testing the demand order on degenerate risks forces strict monotonicity. Once strict monotonicity is in hand, the argument's central object is the following marginal object: after normalizing the safe return to one and writing \(Y=R/s\), the marginal \(\log\)-growth gain from shifting wealth away from the safe asset is
\[\frac{y-1}{1+\alpha(y-1)}.\]
Universal safe-demand monotonicity is equivalent to a single-crossing property of this score: at every portfolio share, the transformed score amplifies bad realizations, \(y<1\), at least as much as good ones, \(y>1\). The universal implications of this single-crossing condition are exactly concavity of \(f\) and convexity of \(r/f(r)\).

As a byproduct of my main theorem, I use the equivalence between the growth-optimal portfolio problem and the rational-inattention problem subject to Shannon costs identified by \citet{RobsonSamuelsonSteiner2023} to characterize comparative risk aversion for a rationally-inattentive agent. After transforming her utility via \(\varphi\), I say such an agent is more risk averse if the induced marginal probability of the safe action is weakly higher in every binary safe/risky problem. For a smooth utility transformation \(\varphi\),

\medskip

\noindent \Cref{cor:ri-c2-growth-rate-transform}. A transformation makes the safe choice more frequent if and only if
\[\varphi \text{ is strictly increasing}
        \qquad\text{and}\qquad
        \varphi''(u)\le-\left|\varphi'(u)\left(1-\varphi'(u)\right)\right|.\]
Ordinary concavity is not enough. We need the new utility to be sufficiently more concave.

Finally, motivated by \citet{curello2025outside}, I study \textit{ex-post} outside options: after an inside asset pays \(r\), the investor can instead take an outside option \(b\) (which is the realization of an \textit{ex-ante} random variable \(B\)). I show that, under a mild support condition, such outside options always make the investor weakly more willing to hold the risky asset. The same conclusion holds in the rational-inattention version, raising the probability of the risky action.

\smallskip

\noindent \textbf{Roadmap.} \secref{sec:score} introduces the model. \secref{sec:one-risky} proves \Cref{thm:main-deterministic}. \secref{sec:ri} provides the corollary for the rationally-inattentive agent. \secref{sec:random} investigates random transformations and contains the outside-option results. I discuss and place my findings in \secref{sec:discussion}. Omitted proofs lie in \Cref{app:omit}.

\section{Setup}\label{sec:score}

An investor faces two assets, one risk-free (deterministic), with a certain payoff of \(s>0\), and one risky. Formally, the risky asset is a strictly positive, finite-support random variable \(R\) distributed according to a full-support probability mass function (pmf) \(p\). Write \(Y\coloneqq R/s\) for the risky return relative to the safe return.

There is an investor who solves the (normalized) \citet{Kelly1956} portfolio problem
\[\max_{\alpha\in\left[0,1\right]} \E_p\log\left(1+\alpha\left(Y-1\right)\right).
\]
This objective is concave and differentiable, with derivative
\[
        \E_p\gamma_\alpha(Y),
        \qquad
        \text{where }
        \gamma_\alpha(y)\coloneqq\frac{y-1}{1+\alpha\left(y-1\right)}.
\]
The (normalized) demand correspondence is
\[
        \D(Y)\coloneqq
        \argmax_{\alpha\in\left[0,1\right]}
        \E_p\log\left(1+\alpha\left(Y-1\right)\right),
\]
which is the risky-share correspondence (after normalization).

A deterministic return transformation \(f\colon\left(0,\infty\right)\to\left(0,\infty\right)\) induces, at safe return \(s\), the relative return transformation
\[
        T_s^f(y)\coloneqq\frac{f(sy)}{f(s)}.
\]
The transformed risky-demand correspondence is \(\D\left(T_s^f(Y)\right)\).

Observe that in the normalized problem, if \(Y\not\equiv1\), then
\[
        \frac{d^2}{d\alpha^2}\E_p\log\left(1+\alpha\left(Y-1\right)\right)
        =
        -\E_p\left[
        \frac{\left(Y-1\right)^2}
        {\left(1+\alpha\left(Y-1\right)\right)^2}
        \right]<0,
\]
so \(\D(Y)\) is a singleton. When \(Y\not\equiv1\), we identify \(\D(Y)\) with its unique element.

For a nonempty correspondence \(A\subseteq\left[0,1\right]\) and scalar \(x\), write \(A\leq x\) to mean \(\sup A\leq x\), and \(A\geq x\) to mean \(\inf A\geq x\).

\begin{definition}
A deterministic transformation \(f\) \emph{raises safe demand universally} if \(\sup \D\left(T_s^f(Y)\right)\leq\D(Y)\) for all \(s > 0\), for all \(Y \not\equiv 1\). It \emph{raises risky demand universally} if \(\inf \D\left(T_s^f(Y)\right)\geq\D(Y)\) for all \(s > 0\), for all \(Y \not\equiv 1\).
\end{definition}

\section{Deterministic Characterization}\label{sec:one-risky}

For \(y\neq1\), define the \emph{score amplification ratio}
\[k_{\alpha,s}^f(y)\coloneqq \frac{\gamma_\alpha(T_s^f(y))}{\gamma_\alpha(y)}.
\]

\begin{lemma}
\label{lem:score-single-crossing}
Let \(f\colon\left(0,\infty\right)\to\left(0,\infty\right)\) be strictly increasing. Then \(f\) raises safe demand universally if and only if \(k_{\alpha,s}^f(a)\ge k_{\alpha,s}^f(b)\) for all \(a < 1 < b\), for all \(\alpha \in \left[0,1\right]\), and for all \(s > 0\). It raises risky demand universally if and only if the reverse inequality holds for all such \(a,b,\alpha,s\).
\end{lemma}

The single-crossing condition says the transformation strengthens bad-realization reasons \textit{against} risk at least as much as good-realization reasons \textit{for} risk; therefore, at the original optimum, the transformed marginal incentive to add risk is still nonpositive, so the transformed optimum cannot move to the right.

Moreover, it turns out that we need only check extreme portfolios:
\begin{lemma}
\label{lem:endpoint-reduction}
For fixed \(a<1<b\) and \(s>0\), the inequality \(k_{\alpha,s}^f(a)\ge k_{\alpha,s}^f(b)\) holds for every \(\alpha\in \left[0,1\right]\) if and only if it holds at \(\alpha=0\) and \(\alpha=1\). The same is true with the reverse inequality.
\end{lemma}

Combining these lemmata, proved in \Cref{app:firstlems}, we have:

\begin{theorem}
\label{thm:main-deterministic}
The function \(f\colon\left(0,\infty\right)\to\left(0,\infty\right)\) raises safe demand universally if and only if \(f\) is strictly increasing, concave, and \(r/f(r)\) is convex. It raises risky demand universally if and only if \(f\) is strictly increasing, convex, and \(r/f(r)\) is concave.
\end{theorem}

\begin{proof}
We prove the raising-safe-demand statement. 

Suppose \(f\) raises safe demand universally. It is easy to see that \(f\) must be strictly increasing: else, we could find some deterministic \(r < s\) with \(f(r)\ge f(s)\). Take safe return \(s\) and deterministic risky return \(R\equiv r\), so \(Y\equiv r/s<1\) and \(\D(Y)=0\). But \(T_s^f(Y)=\frac{f(r)}{f(s)}\ge1\), so \(\sup\D\left(T_s^f(Y)\right)=1>0=\D(Y)\), contradicting universal safe-demand increase. We conclude that \(f\) is strictly increasing. 

\Cref{lem:score-single-crossing} now applies. Moreover, by \Cref{lem:endpoint-reduction}, it suffices to only analyze the endpoints. At \(\alpha=0\),
\[
        k_{0,s}^f(y)=\frac{T_s^f(y)-1}{y-1}
        =\frac{s}{f(s)}\frac{f(sy)-f(s)}{sy-s}.
\]
For \(r<s<t\), take \(a=r/s\) and \(b=t/s\). Then,
\[k_{0,s}^f(a)\ge k_{0,s}^f(b) \qquad \iff \qquad \frac{f(s)-f(r)}{s-r} \ge \frac{f(t)-f(s)}{t-s},\]
which is the three-point slope characterization of concavity of \(f\) \citep[Theorem 1: d \(\iff\) e]{Pratt1964}. Thus, the \(\alpha=0\) endpoint is equivalent to concavity of \(f\).

At \(\alpha=1\), let \(q(r)=r/f(r)\). Since
\[k_{1,s}^f(y) = \frac{y}{T_s^f(y)}\frac{T_s^f(y)-1}{y-1} = 1-\frac{s}{q(s)}\frac{q(sy)-q(s)}{sy-s},\]
\[k_{1,s}^f(a)\ge k_{1,s}^f(b)  \qquad \iff \qquad \frac{q(s)-q(r)}{s-r} \le \frac{q(t)-q(s)}{t-s},\]
which is convexity of \(q\), also thanks to \citet[Theorem 1: d \(\iff\) e]{Pratt1964}. In sum, universal safe-demand increase implies strict increase of \(f\), concavity of \(f\), and convexity of \(r/f(r)\). Conversely, if these three properties hold, the two endpoint inequalities hold; \Cref{lem:endpoint-reduction} delivers the single-crossing condition for every \(\alpha\); and \Cref{lem:score-single-crossing}, the universal safe-demand increase.

The risky-demand statement is the mirror image.
\end{proof}

If \(f\) is \(C^2\), \Cref{thm:main-deterministic} says that safe-demand monotonicity is equivalent to strict increase of \(f\) together with \(f''(r)\le0\) and \(\left(\frac{r}{f(r)}\right)''\ge0\). For a \(C^2\) and strictly-increasing \(f\), in log-return coordinates \(\Phi(x)\coloneqq\log f(\mathrm e^x)\), the two curvature inequalities become
\[
        \Phi''(x)\le-
        \left|\Phi'(x)(1-\Phi'(x))\right|.
\]

An example of a transformation that robustly raises safe demand is \(f(r)=r/(1+r)\). \(f\) is increasing and concave, while \(r/f(r)=1+r\) is affine, hence, convex. A transformation that robustly raises risky demand is \(f(r)=r+c\) with \(c>0\). \(f\) is increasing and affine, hence, convex, while \(r/f(r)=r/(r+c)\) is concave.

To see that mere concavity of the transformation is not enough, let \(s=1\) and let the risky return be \(R=1/2\) or \(R=4\) with equal probabilities. The original risky share is \(5/6\). The transformation \(f(r)=\sqrt r\) is concave, but \(r/f(r)=\sqrt r\) is concave, not convex. After transformation, the relative risky returns are \(\sqrt{1/2}\) and \(2\), and the optimum is \(1\), i.e., the safe share falls.

\section{Risk-Aversion for the Rationally Inattentive}\label{sec:ri} 

\citet[Proposition 2]{RobsonSamuelsonSteiner2023} reveal that the growth-optimal portfolio problem is equivalent to the classic rational-inattention one. In their notation, there is a finite state-space \(\Theta\), with at least two elements. Action \(a\in A\) yields state-dependent utility \(u(a,\theta)\), thus, gross return \(\mathrm e^{u(a,\theta)}\). An optimal portfolio \(\alpha\in\Delta(A)\) solves
\[
        \max_{\alpha\in\Delta(A)}
        \E_p\log\left(\E_\alpha \mathrm e^{u(a,\theta)}\right),
\]
and the associated rational-inattention problem is
\[
        \max_{q\in\Delta(A)^\Theta}
        \left\{
        \E_{p,q}u(a,\theta)-I_{p,q}(a;\theta)
        \right\},
\]
where \(q(a\mid\theta)\) is a state-contingent stochastic choice rule, \(\E_{p,q}\) denotes the expectation under \(p(\theta)q(a\mid\theta)\), and \(I_{p,q}(a;\theta)\) is the mutual information between actions and states. \citet{RobsonSamuelsonSteiner2023} show that a marginal action distribution\footnote{We say that a marginal action distribution \(\alpha\in\Delta(A)\) is induced by a solution of the rational-inattention problem if \(\alpha(a)=\E_p q(a\mid\theta)\) for some optimal rule \(q\).} \(\alpha\) is induced by a solution of the rational-inattention problem if and only if it solves the growth-optimal portfolio problem above.

Again we specialize to two actions, with one risk free or ``safe.'' The safe action \(0\) has state-independent growth rate \(u_0\), and the risky action \(1\) has state-dependent growth rate \(u_R(\theta)\). Write \(s\coloneqq \mathrm e^{u_0}\), and \(R_\theta\coloneqq \mathrm e^{u_R(\theta)}\). For a growth-rate transformation \(\varphi\colon\mathbb{R}\to\mathbb{R}\), define the set of induced marginal probabilities of the risky action by
\[
        \D_\varphi^{RI}(u_0,u_R)
        \coloneqq
        \argmax_{\alpha\in\left[0,1\right]}
        \E_p\log\left((1-\alpha)e^{\varphi(u_0)}+\alpha \mathrm e^{\varphi(u_R(\theta))}\right).
\]
Let \(\iota(u)\coloneqq u\). For \(u_R\not\equiv u_0\), \(\D_\iota^{RI}(u_0,u_R)\) is singleton; so, as before, we identify it with its unique element.

\begin{definition}
A deterministic transformation \(\varphi\) \emph{raises safe action demand universally} if \(\sup \D_\varphi^{RI}(u_0,u_R)\le \D_\iota^{RI}(u_0,u_R)\) for all \(u_0 \in \mathbb{R}\), \(u_R\not\equiv u_0\). It \emph{raises risky action demand universally} if \(\inf \D_\varphi^{RI}(u_0,u_R)\ge \D_\iota^{RI}(u_0,u_R)\) for all \(u_0 \in \mathbb{R}\), \(u_R\not\equiv u_0\).
\end{definition}

\begin{corollary}
\label{cor:ri-growth-rate-transform} Transformation \(\varphi\) raises (lowers) safe action demand universally if and only if \(r \mapsto \exp\left(\varphi\left(\log r\right)\right)\) is strictly increasing and concave (convex) and \(r \mapsto r\exp\left(-\varphi\left(\log r\right)\right)\) is convex (concave).
\end{corollary}

\begin{proof}
The growth-rate transformation \(u\mapsto\varphi(u)\) is exactly the gross-return transformation \(r\mapsto g_\varphi(r)\), so apply \Cref{thm:main-deterministic}.
\end{proof}

\begin{corollary}
\label{cor:ri-c2-growth-rate-transform}
Suppose \(\varphi\in C^2\). Then \(\varphi\) raises safe action demand universally if and only if \(\varphi\) is strictly increasing and
\[
        \varphi''(u)\le-\left|\varphi'(u)\left(1-\varphi'(u)\right)\right|
        \qquad
        \forall u\in\mathbb{R}.
\]
It raises risky action demand universally if and only if \(\varphi\) is strictly increasing and
\[
        \varphi''(u)\ge\left|\varphi'(u)\left(1-\varphi'(u)\right)\right|
        \qquad
        \forall u\in\mathbb{R}.
\]
\end{corollary}
We conduct the straightforward differentiation needed to prove this corollary in \Cref{ap:cor43proof}. Note that even linear transformations fail this criterion, unless they have slope \(1\). This is not paradoxical: we are not altering the cost function in the rational-inattention problem, so the ``scale'' of \(u\) matters.\footnote{Indeed, scaling utilities is equivalent to scaling  the information-acquisition cost, which can obviously alter the resulting stochastic choice rule.}

\section{Random Transformations}\label{sec:random} 

Let \(Z\) be an exogenous random variable that is independent of asset returns, distributed according to \(\mu\). A random transformation \(F = \left(F_z\right)_z\) is a family \(F_z\colon\left(0,\infty\right)\to\left(0,\infty\right)\) that is weakly increasing for \(\mu\)-almost every realization. Assume \((z,r)\mapsto F_z(r)\) is measurable and that all expectations below are finite, with differentiation under the integral valid. 

At safe return \(s\), the random transformation induces the relative-return map
\[T_{z,s}(y)\coloneqq\frac{F_z(sy)}{F_z(s)}.\]
The transformed normalized objective is, thus,
\[G_F(\alpha;s,Y)\coloneqq\E_{p,\mu}\log\left(1+\alpha\left(T_{Z,s}(Y)-1\right)\right),
\]
and the transformed risky-demand correspondence is
\[\D_F(s,Y)\coloneqq\argmax_{\alpha\in\left[0,1\right]}G_F(\alpha;s,Y).\]
The derivative of the transformed objective is
\(G_F'(\alpha;s,Y)=\E_p h_{\alpha,s}^{F}(Y)\),
where \(h_{\alpha,s}^{F}(y)\coloneqq\E_\mu\gamma_\alpha\left(T_{Z,s}(y)\right)\).

For \(y\neq1\), define the \textit{random score amplification ratio} \(k_{\alpha,s}^{F}(y)\coloneqq\frac{h_{\alpha,s}^{F}(y)}{\gamma_\alpha(y)}\). Because \(F_z\) is weakly increasing a.s. and \(T_{z,s}(1)=1\), the transformed score has the same sign as the original score. Hence, \(k_{\alpha,s}^{F}(y)\ge0\) for \(y\neq1\).

\begin{definition} The random transformation \((F_z)_z\) \emph{raises safe demand universally} if \(\sup\D_F(s,Y)\le\D(Y)\) for all \(s>0\) and \(Y\not\equiv1\). It \emph{raises risky demand universally} if \(\inf\D_F(s,Y) \geq \D(Y)\) for all \(s>0\) and \(Y\not\equiv1\).
\end{definition}

We say that \(F\) has \emph{no collapse} if \(\mathbb{P}_\mu\left(F_Z(r)<F_Z(t)\right)>0\) for all \(0<r<t\).

\begin{lemma}
\label{lem:random-no-collapse}
The following are equivalent: \(\textnormal{(i)}\) \(F\) has no collapse; \(\textnormal{(ii)}\) \(k_{\alpha,s}^{F}(a)>0\) for all \(a<1\), \(\alpha\in\left[0,1\right]\), and \(s>0\); and \(\textnormal{(iii)}\) \(k_{\alpha,s}^{F}(b)>0\) for all \(b>1\), \(\alpha\in\left[0,1\right]\), and \(s>0\).

If \(F\) has no collapse, then \(\D_F(s,Y)\) is singleton for every \(s>0\) and every \(Y\not\equiv1\).
\end{lemma}

What this lemma says is that we can tolerate random transformations that flatten payoff differences sometimes, but not ones that erase any strict payoff comparison almost surely. As long as every strict comparison survives with positive probability, bad outcomes still push against risk, good outcomes still push toward risk, and the transformed objective has a unique optimum. The proof of this lemma, as well as the theorem that follows (\Cref{thm:random-scalar}), lie in \Cref{app:sec5proofs}.

\begin{theorem}
\label{thm:random-scalar}
\(F\) raises safe (risky) demand universally if and only if it has no collapse and, for every \(\alpha\in\left[0,1\right]\) and \(s>0\), \(k_{\alpha,s}^{F}(a)\ge \ (\leq) \ k_{\alpha,s}^{F}(b)\) for all \(a < 1 < b\).
\end{theorem}

This theorem is the random-transform version of \Cref{lem:score-single-crossing}. Again, we need the transformation to strengthen bad-realization reasons \textit{against} risk at least as much as good-realization reasons \textit{for} risk. The extra-ingredient--no collapse--ensures that these marginal forces always have bite.

\subsection{Outside Options}

Consider an independent \textit{ex-post} outside option \(B>0\), applied asset by asset. \textit{Viz.,} if an asset realizes gross return \(r\), the investor receives \(r\vee B\coloneqq\max\left\{r,B\right\}\). Equivalently, this is the random family \(F_B(r)=r\vee B\). In relative-return form,
\[
        T_{B,s}(y)\coloneqq\frac{sy\vee B}{s\vee B}
        =
        \frac{y\vee X}{1\vee X},
        \qquad \text{with} \qquad
        X\coloneqq\frac{B}{s}.
\]
We also assume that for all \(x > 0\), \(\mathbb{P}(B<x)>0\). Then, 

\begin{corollary}
\label{cor:outside-risky-demand}
An outside option raises risky demand universally.
\end{corollary}

\begin{proof}
We verify the hypotheses of \Cref{thm:random-scalar}. By assumption, \(F\) has no collapse. We verify the risky-demand single-crossing condition. Fix \(s>0\), write \(X=B/s\), and fix \(a<1<b\). For a realization \(X=x\), if \(x\le1\), then \(T_{x,s}(b)=b\) and \(T_{x,s}(a)\in\left[a,1\right]\). Since \(\gamma_\alpha\) is increasing, negative below one, and zero at one,
\[
        0\le
        \frac{\gamma_\alpha\left(T_{x,s}(a)\right)}{\gamma_\alpha(a)}
        \le1
        =
        \frac{\gamma_\alpha\left(T_{x,s}(b)\right)}{\gamma_\alpha(b)}.
\]
If \(x>1\), then \(T_{x,s}(a)=1\) and \(T_{x,s}(b)\in\left[1,b\right]\), so
\[
        \frac{\gamma_\alpha\left(T_{x,s}(a)\right)}{\gamma_\alpha(a)}
        =0
        \le
        \frac{\gamma_\alpha\left(T_{x,s}(b)\right)}{\gamma_\alpha(b)}.
\]
Therefore, for every realization \(x\),
\[
        \frac{\gamma_\alpha\left(T_{x,s}(a)\right)}{\gamma_\alpha(a)}
        \le
        \frac{\gamma_\alpha\left(T_{x,s}(b)\right)}{\gamma_\alpha(b)}.
\]
Averaging over \(X\) yields
\[
        k_{\alpha,s}^{F}(a)\le k_{\alpha,s}^{F}(b)
        \qquad
        \forall a<1<b,\ \forall\alpha\in\left[0,1\right],\ \forall s>0.
\]
By \Cref{thm:random-scalar}, \(F_B\) raises risky demand universally.
\end{proof}

\subsection{Outside Options for the Rationally Inattentive}

In fact, a random \textit{ex-post} outside option also makes a rationally inattentive agent less risk averse. Let \(B\) be a real-valued random variable, independent of \(\theta\), suppose the relevant expectations are finite, and suppose \(\mathbb{P}(B<v)>0\) for all \(v\in\mathbb{R}\).

For \(u_0\in\mathbb{R}\) and \(u_R\not\equiv u_0\), define
\[
        \D_B^{RI}(u_0,u_R)\coloneqq
        \argmax_{\alpha\in\left[0,1\right]}
        \E_{\theta,B}\log\left((1-\alpha)e^{u_0\vee B}+\alpha \mathrm e^{u_R(\theta)\vee B}\right).
\]
Let \(\D_\iota^{RI}(u_0,u_R)\) denote the corresponding untransformed demand correspondence:
\[
        \D_\iota^{RI}(u_0,u_R)\coloneqq
        \argmax_{\alpha\in\left[0,1\right]}
        \E_\theta\log\left((1-\alpha)e^{u_0}+\alpha \mathrm e^{u_R(\theta)}\right).
\]
\begin{corollary}
\label{cor:ri-random-growth-floor}
An outside option raises risky-action demand universally:
\(\D_B^{RI}(u_0,u_R)\) is singleton, and \(\D_B^{RI}(u_0,u_R)\ge\D_\iota^{RI}(u_0,u_R)\) for all \(u_0\in\mathbb{R}\), for all \(\ u_R\not\equiv u_0\).
\end{corollary}

\begin{proof}
The transformation \(u\mapsto u\vee B\) induces the gross-return transformation \(r\mapsto \mathrm e^{\log r\vee B}=r\vee \mathrm e^B\). Let \(C\coloneqq \mathrm e^B\). Our support condition implies \(\mathbb{P}(C<t)=\mathbb{P}(B<\log t)>0\) for all \(t>0\), so we merely apply \Cref{cor:outside-risky-demand}.
\end{proof}

\section{Discussion}\label{sec:discussion}

Classical comparative risk aversion transforms Bernoulli utilities: the more risk-averse utility is an increasing concave transformation of the less risk-averse one. Stronger orders ask for robustness to richer comparisons, such as partial insurance and background risk. Prudence moves the transformation logic to marginal utility. I move the transformation elsewhere: which return scales make the same Kelly investor behave more conservatively in every one-safe-one-risky growth problem? 

I show that return transformations are not utility transformations in disguise. Although ordinary concavity of \(f\) is necessary for universal safe-demand increases, it is not sufficient. The transformation must also satisfy the harmonic curvature condition that
\(r\mapsto r/f(r)\) is convex, i.e., it must not only compress upside in the usual sense but also make original returns sufficiently convex when measured in transformed-return units. 

The paper is related to, but distinct from, the portfolio-dominance literature,\footnote{See, e.g., \citet{Gollier1995}, \citet{Gollier1997}, and \citet{Gollier2001}. \citet{DybvigWang2012} emphasize that robust portfolio-share comparative statics are fragile and instead study comparative statics for the distribution of portfolio payoffs.} which changes the distribution of a risky payoff and asks whether all agents in some preference class move in the same direction. Here the distributional environment is fixed. The transformation acts on every gross payoff, including the safe return, generating a new return scale on which the entire portfolio problem is evaluated. 
My exercise also differs from the standard portfolio-growth questions: I ask not which strategy grows fastest, and not how information changes the growth rate, but which cardinal transformations of returns move the safe share in a uniform direction.\footnote{See \citet{Kelly1956}, \citet{Breiman1961}, \citet{AlgoetCover1988}, \citet{Cover1991}, \citet{CoverOrdentlich1996}, and \citet{CoverThomas2006}. \citet{Samuelson1979} is a classic warning against treating log growth as a universal preference criterion, and \citet{MacLeanThorpZiemba2011} collect modern treatments of the Kelly criterion.}

\appendix

\section{Omitted Proofs}\label[appendix]{app:omit}

\subsection{Proof of \texorpdfstring{\Cref{lem:score-single-crossing,lem:endpoint-reduction}}{Lemmas}}\label[appendix]{app:firstlems}

\begin{proof}[Proof of \Cref{lem:score-single-crossing}]
We prove the safe-demand statement. In relative-return units, the safe asset is \(1\). By definition, \(T_s^f(1)=f(s)/f(s)=1\). Since \(f\) is strictly increasing, \(T_s^f(y)\lessgtr 1\) exactly when \(sy\lessgtr s\), equivalently, when \(y\lessgtr 1\).

Suppose the displayed single-crossing condition holds. Let \(Y\not\equiv1\), and let \(\alpha^*\) and \(\alpha_f^*\) be the unique maximizers of the original and transformed problems. If \(\alpha^*=1\), then \(\alpha_f^*\le1=\alpha^*\) (as no risky share can exceed \(1\)). 

If the support of \(Y\) lies weakly below one, then both optima are zero, so assume there are realizations on both sides of one. For \(\alpha^*<1\), concavity of the original objective implies \(\E_{p}\gamma_{\alpha^*}(Y)\le0\). Choose \(\lambda\) between the supremum of \(k_{\alpha^*,s}^f\) on the good realizations \(\{y>1\}\) and the infimum on the bad realizations \(\{y<1\}\). We may ignore the realization \(y=1\) as it contributes zero to both marginal scores. Since \(\gamma_\alpha(y)<0\) below one and \(\gamma_\alpha(y)>0\) above one,
\[
        k_{\alpha^*,s}^f(y)\gamma_{\alpha^*}(y)
        \le
        \lambda\gamma_{\alpha^*}(y)
        \qquad
        \forall y.
\]
Taking expectations:
\[
        \E_{p}\gamma_{\alpha^*}(T_s^f(Y))
        \le
        \lambda\E_{p}\gamma_{\alpha^*}(Y)
        \le0.
\]
Thus, the transformed objective is nonincreasing at \(\alpha^*\), so \(\alpha_f^*\le\alpha^*\).

Conversely, suppose the single-crossing inequality fails for some \(a<1<b\), \(s>0\), and \(\alpha_0\in[0,1]\):
\(k_{\alpha_0,s}^f(a)<k_{\alpha_0,s}^f(b)\). If \(\alpha_0=1\), then the function
\[\alpha\mapsto k_{\alpha,s}^f(a)-k_{\alpha,s}^f(b)\]
is continuous on \([0,1]\), because \(T_s^f(a)\) and \(T_s^f(b)\) are fixed positive constants and the denominators in \(k_{\alpha,s}^f\) are positive. Hence, the same strict failure holds for some \(\alpha<1\) close to one. If \(\alpha_0<1\), set \(\alpha=\alpha_0\). Choose probabilities on \(\{a,b\}\) such that
\(\E_{p}\gamma_\alpha(Y)=0\), so that  \(\alpha\) is the original optimum. Since \(\gamma_\alpha(a)<0<\gamma_\alpha(b)\) and \(k_{\alpha,s}^f(a)<k_{\alpha,s}^f(b)\),
\[
        \E_{p}\gamma_\alpha(T_s^f(Y))
        =
        \E_{p}\left[k_{\alpha,s}^f(Y)\gamma_\alpha(Y)\right]
        >0.
\]
The transformed optimum, therefore, lies strictly to the right of \(\alpha\), contradicting universal safe-demand increase.

The risky-demand proof is basically identical (\textit{mutatis mutandis}).\end{proof}

\begin{proof}[Proof of \Cref{lem:endpoint-reduction}]
Write \(T=T_s^f\). Since
\[
        k_{\alpha,s}^f(y)=
        \frac{T(y)-1}{y-1}
        \frac{1+\alpha(y-1)}{1+\alpha(T(y)-1)},
\]
the inequality \(k_{\alpha,s}^f(a)\ge k_{\alpha,s}^f(b)\) becomes
\[
        \frac{T(a)-1}{a-1}(1+\alpha(a-1))(1+\alpha(T(b)-1))
        \ge
        \frac{T(b)-1}{b-1}(1+\alpha(b-1))(1+\alpha(T(a)-1)).
\]
The difference is affine in \(\alpha\).\end{proof}

\subsection{Proof of \texorpdfstring{\Cref{cor:ri-c2-growth-rate-transform}}{Corollary}}\label[appendix]{ap:cor43proof}
\begin{proof}[Proof of \Cref{cor:ri-c2-growth-rate-transform}]
Let \(r= \mathrm{e}^u\) and \(g_\varphi(r) \coloneqq \mathrm{e}^{\varphi(u)}\). Then
\[
        g_\varphi''(r)=
        \frac{e^{\varphi(u)}}{r^2}
        \left[
        \varphi''(u)+\varphi'(u)^2-\varphi'(u)
        \right].
\]
Thus, \(g_\varphi\) is concave if and only if \(\varphi''(u)\le\varphi'(u)\left(1-\varphi'(u)\right)\),
and \(g_\varphi\) is convex if and only if the reverse inequality holds.

A similar direct differentiation tells us that \(r/g_\varphi(r)\) is convex if and only if \(\varphi''(u)\le-\varphi'(u)\left(1-\varphi'(u)\right)\), and \(r/g_\varphi(r)\) is concave if and only if the reverse inequality holds. 

We combine the two safe-demand inequalities to get
\[\varphi''(u)\le
        \min\left\{
        \varphi'(u)\left(1-\varphi'(u)\right),
        -\varphi'(u)\left(1-\varphi'(u)\right)
        \right\}
        =
        -\left|\varphi'(u)\left(1-\varphi'(u)\right)\right|.
\]
Likewise, for the two risky-demand inequalities. Finally, strict increase of \(\varphi\) is exactly strict increase of \(g_\varphi\).
\end{proof}

\subsection{Proofs of \texorpdfstring{\secref{sec:random}}{section} Results}\label[appendix]{app:sec5proofs}

We prove \Cref{lem:random-no-collapse} and \Cref{thm:random-scalar}.
\begin{proof}[Proof of \Cref{lem:random-no-collapse}]
For \(a<1\), monotonicity implies \(T_{z,s}(a)\le1\) a.s. Thus, \(\gamma_\alpha\left(T_{z,s}(a)\right)\le0\) a.s., with equality if and only if \(T_{z,s}(a)=1\). Since \(\gamma_\alpha(a)<0\),
\[
        k_{\alpha,s}^{F}(a)>0 \quad \iff \quad
        h_{\alpha,s}^{F}(a)<0
        \quad \iff \quad
        \mathbb{P}_\mu\left(T_{Z,s}(a)<1\right)>0
        \quad \iff \quad
        \mathbb{P}_\mu\left(F_Z(sa)<F_Z(s)\right)>0.
\]
Since \(s>0\) is arbitrary, this is exactly \(F\) having no collapse. The argument for \(b>1\) is symmetric.

Now suppose \(F\) has no collapse and \(Y \not\equiv 1\), so that
\(\mathbb P_{p,\mu}\left(T_{z,s}(Y)\neq 1\right)>0\). For each realization of \((Y,Z)\), the function
\[\alpha \mapsto \log\left(1+\alpha\left(T_{z,s}(Y)-1\right)\right)\]
is concave on \([0,1]\), and it is strictly concave whenever \(T_{z,s}(Y)\neq 1\). Hence, for any distinct
\(\alpha_0,\alpha_1\in[0,1]\) and any \(\lambda\in(0,1)\), taking
expectations:
\[
G_F\left(\lambda\alpha_0+(1-\lambda)\alpha_1;s,Y\right)
>
\lambda G_F(\alpha_0;s,Y)
+
(1-\lambda)G_F(\alpha_1;s,Y),
\]
i.e., \(G_F(\cdot;s,Y)\) is strictly concave, so \(\D_F(s,Y)\) is singleton.
\end{proof}

\begin{proof}[Proof of \Cref{thm:random-scalar}]
We prove the safe-demand statement. Suppose \(F\) has no collapse and \(k_{\alpha,s}^{F}(a)\ge k_{\alpha,s}^{F}(b)\) for every \(a<1<b\), \(\alpha\in\left[0,1\right]\), and \(s>0\). By \Cref{lem:random-no-collapse}, \(\D_F(s,Y)\) is singleton for every \(s>0\) and every \(Y\not\equiv1\).

Fix \(s>0\) and \(Y\not\equiv1\), and let \(\alpha\) be the unique element of \(\D(Y)\). If \(\alpha=1\), then every transformed maximizer is weakly below \(\alpha\). Hence, suppose \(\alpha<1\). By original optimality, \(\E\gamma_\alpha(Y)\le0\).

We claim that \(\E h_{\alpha,s}^{F}(Y)\le0\). If there are no bad realizations \(y<1\), then \(\E\gamma_\alpha(Y)\le0\) forces \(Y\equiv1\), which is excluded. If there are bad realizations but no good ones, then every transformed score is nonpositive, so \(\E h_{\alpha,s}^{F}(Y)\le0\).

Now suppose the support of \(Y\) has realizations on both sides of one. Choose \(\lambda > 0\) such that
\[
        \sup_{y>1}k_{\alpha,s}^{F}(y)\le\lambda\le
        \inf_{y<1}k_{\alpha,s}^{F}(y),
\]
ignoring the realization \(y=1\). Since \(\gamma_\alpha(y)<0\) below one and \(\gamma_\alpha(y)>0\) above one,
\[
        k_{\alpha,s}^{F}(y)\gamma_\alpha(y)\le\lambda\gamma_\alpha(y)
        \qquad
        \forall y\neq1
\]
on the support of \(Y\). The realization \(y=1\) contributes zero to both marginal scores. Therefore,
\[
        \E_p h_{\alpha,s}^{F}(Y)\le\lambda\E_p\gamma_\alpha(Y)\le0.
\]
Thus, \(G_F'(\alpha;s,Y)\le0\). Since \(G_F(\cdot;s,Y)\) is strictly concave, its unique maximizer lies weakly to the left of \(\alpha\). Hence, \(\sup \D_F(s,Y)\le \D(Y)\).

Conversely, suppose \(F\) raises safe demand universally. We first prove \(F\) has no collapse. If no collapse fails, weak monotonicity yields \(F_z(r)=F_z(t)\) a.s. for some \(0<r<t\). Set \(s=t\) and \(Y\equiv r/t<1\). Then \(\D(Y)=\left\{0\right\}\), while \(T_{z,s}(Y)=1\) a.s., so \(\D_F(s,Y)=\left[0,1\right]\). Hence, \(\sup\D_F(s,Y)=1>0= \D(Y)\), contradicting universal safe-demand increase.

It remains to prove the single-crossing inequality. If it fails, then for some \(s>0\), \(\alpha\in\left[0,1\right]\), and \(a<1<b\), \(k_{\alpha,s}^{F}(a)<k_{\alpha,s}^{F}(b)\).
Choose probabilities on \(\left\{a,b\right\}\) so that \(\E_p \gamma_\alpha(Y)=0\). Then \(\alpha\in\D(Y)\) and \(\E_p h_{\alpha,s}^{F}(Y)>0\). If \(\alpha<1\), strict concavity of the transformed objective gives it a unique maximizer strictly above \(\alpha\), contradicting universal safe-demand increase. If \(\alpha=1\), perturb the same two-point probabilities slightly toward the bad realization. Then \(\E_p \gamma_1(Y)<0\) while \(\E_p h_{1,s}^{F}(Y)>0\). The original maximizer lies below one, while the transformed objective is maximized at one, again contradicting universal safe-demand increase. Therefore, \(k_{\alpha,s}^{F}(a)\ge k_{\alpha,s}^{F}(b)\) for every \(a<1<b\), \(\alpha\in\left[0,1\right]\), and \(s>0\).

The risky-demand proof is the obvious mirror, \textit{mutatis mutandis}.
\end{proof}

\bibliography{sample}

\end{document}